\documentclass[a4paper,12pt]{article}
\usepackage{amsmath,amsthm,amssymb}
\usepackage{graphicx}
\usepackage{booktabs}
\usepackage{hyperref}
\usepackage{enumitem}
\setlist[itemize]{topsep=0pt,partopsep=0pt,itemsep=0pt,parsep=0pt,leftmargin=2em}
\setlist[enumerate]{topsep=0pt,partopsep=0pt,itemsep=0pt,parsep=0pt,leftmargin=2em}
\usepackage[left=2.5cm,right=2.5cm,top=2.5cm,bottom=2.5cm]{geometry}
\usepackage{setspace}
\usepackage{cite}
\usepackage{multirow}
\usepackage{xcolor}
\usepackage{tikz}
\usetikzlibrary{arrows.meta}

\usepackage{titlesec}
\titleformat{\section}{\normalsize\bfseries}{\thesection}{1em}{}
\titlespacing{\section}{0pt}{3ex}{1ex}
\titleformat{\subsection}{\normalsize\bfseries}{\thesubsection}{1em}{}
\titlespacing{\subsection}{0pt}{2.5ex}{0.8ex}
\titleformat{\subsubsection}{\normalsize\bfseries}{\thesubsubsection}{1em}{}
\titlespacing{\subsubsection}{0pt}{2ex}{0.5ex}

\newtheorem{theorem}{Theorem}
\newtheorem{lemma}{Lemma}
\newtheorem{corollary}{Corollary}
\newtheorem{definition}{Definition}
\newtheorem{proposition}{Proposition}
\newtheorem{assumption}{Assumption}

\definecolor{copenblue}{RGB}{0,102,204}
\definecolor{copenlight}{RGB}{173,216,230}
\definecolor{copendark}{RGB}{0,51,102}
\definecolor{copengreen}{RGB}{0,128,0}

\title{\textbf{The Inevitability of Side-Channel Leakage in Encrypted Traffic}\thanks{The Chinese version of this paper has been accepted by \emph{Acta Electronica Sinica}. Please cite the official published version.}}
\author{Guangjie Liu, Guang Cheng, Weiwei Liu}
\date{\today}

\begin{document}
\maketitle

\pagenumbering{Roman}
\setcounter{page}{1}

\begin{abstract}

The widespread adoption of TLS 1.3 and QUIC has rendered payload content invisible, shifting traffic analysis toward reliance on side-channel features. However, rigorous justification for ``why side-channel leakage is inevitable in encrypted communications'' has long been lacking. This paper establishes a strict foundation from information theory and system design by constructing a formal model \(\Sigma=(\Gamma,\Omega)\), where the encrypted communication model \(\Gamma=(A,\Pi,\Phi,N)\) describes the causal chain of ``application generation--protocol encapsulation--encryption transformation--network transmission'', and the observation model \(\Omega\) characterizes external observation capabilities. Based on the composite channel structure, data processing inequality, and stable propagation of bounded Lipschitz statistics, we propose and prove the ``Side-Channel Existence Theorem'': for distinguishable semantic pairs, under the conditions that the system satisfies mapping non-degeneracy (bounded metric expectation \(\mathbb{E}[d(z_P,z_N)\mid X]\le C\)), protocol-layer statistical distinguishability (expectation difference \(\ge\bar\Delta\)), Lipschitz continuity of statistics, observation non-degeneracy (preservation ratio \(\rho>0\)), and the distinguishability propagation condition (\(C<\bar\Delta/2L_\varphi\)), the mutual information \(I(X;Y)\) between observed features and semantic variables is necessarily strictly positive with an explicit lower bound. The corollary demonstrates that in efficiency-prioritized multi-semantic systems, side-channel leakage is inevitable as long as at least one pair of applications is statistically distinguishable. Three key factors jointly determine the leakage boundary: the mapping non-degeneracy constant \(C\) is constrained by efficiency requirements, semantic distinguishability \(\bar\Delta\) stems from application diversity, and observation non-degeneracy \(\rho\) is determined by analyst capabilities. This paper establishes, for the first time, a rigorous information-theoretic foundation for encrypted traffic side channels, providing verifiable predictions for attack feasibility, quantifiable performance benchmarks for defense mechanisms, and mathematical basis for engineering decisions on efficiency-privacy tradeoffs.


\end{abstract}

\newpage

\pagenumbering{arabic}
\setcounter{page}{1}

\section{Introduction}

With the widespread deployment of encryption protocols such as TLS 1.3 and QUIC, payload content in modern network communications has been strongly protected by cryptographic means. However, encrypted traffic analysis—which infers sensitive information (such as websites visited by users, applications used, or types of content transmitted) by observing metadata features of encrypted traffic (packet lengths, timing, direction, etc.)—has still achieved remarkable success: in closed environments, website fingerprinting accuracy can reach 91\%-95\%\cite{wang2014effective}, and application identification accuracy exceeds 90\%\cite{shen2021graphdapp}; for encrypted anonymous traffic identification, in real campus gateway environment tests, precision can reach 96\% under medium base rate scenarios, and remain above 93\% even at extremely low base rates of 1000:1\cite{mei2025high}. This phenomenon raises a fundamental theoretical question: why does traffic analysis remain effective even when using computationally secure encryption algorithms such as AES-256-GCM?

From a cryptographic perspective, encryption algorithms ensure the confidentiality of payload content; without the key, attackers cannot recover plaintext from ciphertext. However, encryption protocols inevitably cannot hide communication metadata: source/destination addresses, packet lengths, timestamps, and other information are necessary for network routing and transmission control. This metadata forms the foundation of side channels. The concept of side channels first appeared in the field of cryptographic hardware implementation. Kocher (1996) proposed timing attacks that recover keys by measuring the execution time of encryption operations\cite{kocher1996timing}, and Kocher et al. (1999) introduced Differential Power Analysis (DPA) that extracts key information using power consumption variations of encryption devices\cite{kocher1999differential}. The common characteristic of these attacks is: exploiting physical byproducts (time, power consumption, electromagnetic radiation, etc.) during the implementation of cryptographic algorithms, rather than directly breaking the algorithms themselves. Subsequently, the side-channel concept expanded to the field of network communications. Hintz's (2002) work demonstrated how to identify websites visited by users by analyzing traffic patterns of encrypted web pages\cite{hintz2002fingerprinting}, promoting in-depth research on network traffic side channels. Network traffic side channels are essentially consistent with hardware side channels—encryption protects payload content, but observable features of the communication process itself (packet size, timing, direction, etc.) become a new source of information leakage. This paper focuses on network traffic side channels but inherits the core insight from hardware side-channel research: side channels are not defects in encryption algorithms, but inevitable byproducts of efficiency-prioritized implementation and deployment processes.

The root cause of network traffic side-channel leakage lies in the ``fingerprint propagation'' of application-layer behavioral patterns during transmission. Specifically: size differences in application objects map to length patterns of encrypted record sequences; the timing logic of HTTP request-response is reflected in packet arrival interval distributions; the directionality and burst characteristics of client-server interactions reflect protocol semantics; layer-by-layer encapsulation in the protocol stack amplifies subtle differences (TCP segmentation boundaries, congestion control window adjustments). These mechanisms make traffic from different applications statistically distinguishable—video streams show periodic bursts, instant messaging manifests as bidirectional interaction with small packets, and web browsing displays short-term high-density resource loading. Although encryption operations change payload content, they cannot eliminate these statistical fingerprints originating from application logic.

However, merely pointing out ``fingerprint propagation'' does not answer why side channels are inevitable. The deeper questions are: under given system constraints (limited computational resources, communication efficiency requirements, protocol compatibility needs), why does encryption system design necessarily lead to the existence of side channels? Does this existence have a quantifiable theoretical lower bound? Can the tradeoff relationship between efficiency and privacy be rigorously characterized mathematically?

In recent years, despite continuous innovation in attack methods for encrypted traffic analysis and significant performance improvements\cite{lin2022etbert,shen2023machine}, systematic justification for the theoretical inevitability of side-channel existence remains lacking. Li et al. (2018) quantified website fingerprinting leakage using mutual information \(I(X;Y)\), finding in experiments targeting the Tor network that even though Tor encrypts communication content and hides endpoint identities through three-layer relays, metadata features still leak website information: information leakage \(I(F;W)\) from individual features reaches up to 3.45 bits, and combined features can reach approximately 6.6 bits\cite{li2018measuring}. This non-zero mutual information stems from Tor not obfuscating or padding traffic metadata to ensure low latency and low overhead. Cai et al. (2014) proved that any defense achieving \(\varepsilon\)-security necessarily incurs a computable bandwidth overhead lower bound in a closed world of \(n\) websites\cite{cai2014systematic}, revealing the theoretical cost of defense, but this lower bound only targets specific scenarios and lacks generality. The differential privacy framework provides quantifiable privacy guarantees for traffic analysis defense\cite{dwork2006distributed,sabzi2024netshaper}, but its design philosophy seeks a balance between privacy and utility, rather than exploring the fundamental reasons for information leakage under efficiency constraints.

This paper provides a strict foundational reasoning for side-channel existence from information theory and system design. We first construct a formal model \(\Sigma=(\Gamma,\Omega)\), where the encrypted communication system \(\Gamma=(A,\Pi,\Phi,N)\) characterizes the processes of application generation, protocol encapsulation, encryption transformation, and network transmission, and the observation model \(\Omega\) delimits the layers and positions of external observation, formalizing ``generation--encapsulation--encryption--transmission--observation'' as a causally measurable composite channel \(X\to\Xi_A\to\Xi_P\to\Xi_C\to\Xi_N\to Y\). Within this framework, we externalize the constraints of efficiency-prioritized design as mapping non-degeneracy: there exists a metric \(d\) and constant \(C<\infty\) such that the trajectory mapping from protocol layer to network layer maintains bounded deviation in the metric sense \(\mathbb{E}[d(z_P,z_N)\mid X]\le C\). This metric captures the joint statistical structure of multiple dimensions such as length, timing, and direction of point processes, and the existence of bounded deviation \(C\) stems from practical requirements such as bandwidth, latency, and throughput. Combined with the stable propagation property of bounded Lipschitz statistics, we establish a strict derivation chain from protocol-layer distinguishability to observation-layer information leakage existence.

The core conclusion of this paper is the ``Side-Channel Existence Theorem'': for distinguishable semantic pairs, when the following conditions are satisfied—the system mapping maintains non-degeneracy (metric expectation bound \(\mathbb{E}[d(z_P,z_N)\mid X]\le C\)), protocol-layer statistical distinguishability exists (expectation difference \(\ge\bar\Delta\)), statistics satisfy Lipschitz continuity, the observation model maintains non-degeneracy (preservation ratio \(\rho>0\)), and the distinguishability propagation condition (\(C<\bar\Delta/2L_\varphi\))—the mutual information between observed features and semantic variables satisfies the explicit lower bound
\begin{equation}
I(X;Y)\ge\frac{1}{2\ln 2}\left(\frac{\rho[\bar\Delta-2L_\varphi C]}{2}\right)^2>0.
\end{equation}
This lower bound reveals the inevitability of side-channel leakage: the mapping non-degeneracy constant \(C\) is limited by efficiency constraints, semantic distinguishability \(\bar\Delta\) originates from application diversity, and observation non-degeneracy \(\rho\) is determined by analyst capabilities—these three factors jointly constitute an insurmountable leakage boundary. The corollary shows that in efficiency-prioritized multi-semantic systems, as long as at least one pair of applications is statistically distinguishable, side-channel leakage is necessarily positive.

Therefore, side channels are not incidental flaws in any protocol implementation, but inherent properties of network communication systems satisfying practical constraints. The correct engineering objective is not to pursue unattainable zero leakage, but rather a constrained optimization problem that minimizes leakage under given efficiency constraints.

This paper is organized as follows: Section 2 reviews related theoretical work; Section 3 constructs a formal model for side-channel analysis, including definitions of encrypted communication systems, observation models, and key properties; Section 4 proposes and proves the side-channel leakage existence theorem, establishing a derivation chain from expectation difference to mutual information; Section 5 discusses theoretical implications, analyzes the transformation from information-theoretic lower bounds to actual attack performance, the fundamentality of efficiency-privacy tradeoffs, and theoretical boundaries of defense; Section 6 concludes the paper, introduces several open problems regarding the practical implementation of the theoretical framework, and provides research outlook.

\section{Related Theoretical Work}

The theoretical foundation for side-channel analysis in networks can be traced back to early research in information theory and anonymity metrics. Chaum (1981) proposed mix networks, laying the foundation for anonymous communication\cite{chaum1981untraceable}, but rigorous mathematical frameworks for quantifying anonymity were long lacking. Díaz et al. (2002) first proposed measuring anonymity using Shannon entropy \(H(X) = -\sum p_i \log p_i\), introducing normalized entropy \(d = H(X)/H_M\) to distinguish the effects of anonymity set size and probability distribution\cite{diaz2002measuring}. The key insight of this work was: probability distribution is more important than set size—uniform distribution among 10 people (\(d=1\)) provides stronger anonymity than 100 people where one person has 90\% identification probability (\(d\approx0.47\)). Subsequent work introduced richer entropy metric tools: Deng et al. (2006) applied Rényi entropy \(H_\alpha(X) = \frac{1}{1-\alpha} \log \sum p_i^\alpha\) to anonymity metrics\cite{deng2006measuring}, where different \(\alpha\) values capture different aspects of the distribution (\(\alpha \to 1\) reduces to Shannon entropy). Serjantov and Danezis (2002) introduced mutual information \(I(X;Y)\) into anonymity protocol analysis\cite{serjantov2002information}. Chatzikokolakis et al. (2007) further interpreted anonymity protocols as noisy channels, characterizing leakage upper bounds with channel capacity \(C = \max I(X;Y)\), and introduced relative entropy \(D(P \| Q) = \sum p_i \log(p_i/q_i)\) to measure changes in distribution before and after attacks\cite{chatzikokolakis2007anonymity}. These early works established the theoretical framework for quantifying privacy leakage using information theory, but mainly focused on anonymous communication protocols (such as mix networks and onion routing), without systematically analyzing the relationship between encryption protocols themselves and side channels.

Research on statistical leakage attacks revealed threats from long-term observation. Kesdogan et al. (2002) proved that in open environments, as long as users have habitual communication patterns, long-term observation inevitably leads to anonymity degradation\cite{kesdogan2002limits}. Mathematically, this is equivalent to conditional entropy \(H(X|Y_1, Y_2, \ldots, Y_t)\) decreasing toward near-zero with the number of observations \(t\), revealing the law of exponential decay of anonymity over time. Danezis (2003) proposed statistical disclosure attacks, identifying sender-receiver pairs in anonymous systems through traffic analysis\cite{danezis2003statistical}. These studies revealed the impact of temporal correlations but did not address the causal relationship between encryption protocol design and side channels—namely, why fundamental constraints in protocol design lead to side-channel existence.

The differential privacy framework provides quantifiable guarantees for privacy protection. Dwork et al. (2006) proposed differential privacy: mechanism \(\mathcal{M}\) satisfies \((\varepsilon,\delta)\)-DP if for adjacent datasets \(D, D'\) and any output set \(S\), \(\Pr[\mathcal{M}(D) \in S] \leq e^{\varepsilon} \cdot \Pr[\mathcal{M}(D') \in S] + \delta\)\cite{dwork2006distributed}. The advantage of this framework lies in providing precise privacy parameters and composability: \(k\) \((\varepsilon,\delta)\)-DP operations provide \((k\varepsilon, k\delta)\)-DP guarantee under basic composition. Vuvuzela (2015) and Stadium (2017) first applied differential privacy to large-scale messaging system metadata protection\cite{vandenhooff2015vuvuzela,tyagi2017stadium}. NetShaper (2024) first established a formal differential privacy framework for network side-channel defense\cite{sabzi2024netshaper}. The key theoretical contribution of this work is connecting abstract \(\varepsilon\) privacy parameters with concrete system performance metrics (bandwidth, latency) through computable mathematical relationships. Specifically, given a privacy budget \(\varepsilon\) and network conditions, NetShaper can compute the minimum bandwidth overhead and latency increment required to achieve that privacy guarantee. This framework transforms the qualitative tradeoff of ``privacy vs. performance'' into an optimizable mathematical problem. However, this approach still assumes ``a certain degree of leakage is acceptable'' (encoded through \(\varepsilon>0\)), rather than proving the inevitability of side channels as byproducts of functional implementation from a strict mathematical perspective.

Research on website fingerprinting attacks and provable defenses provides empirical foundations and theoretical attempts for understanding side channels. Early website fingerprinting attacks from Hintz (2002) on traffic analysis of encrypted web pages\cite{hintz2002fingerprinting} to Panchenko et al. (2011) using support vector machines\cite{panchenko2011website} gradually demonstrated the effectiveness of side channels. Dyer et al. (2012) argued that most ``efficient'' traffic shaping schemes still fail due to observable side channels, proposing the BuFLO baseline defense with fixed rate, constant-length packets, and minimum duration strategies\cite{dyer2012peek}. Cai et al. (2014) established an attack-agnostic evaluation framework and bandwidth lower bound, proposing the Tamaraw defense based on this: separate rate fixing for upstream and downstream with block-based padding, providing provable upper bounds on attack accuracy\cite{cai2014systematic}. Wang and Goldberg (2017) proposed the Walkie-Talkie defense, using half-duplex burst shaping and pairwise obfuscation mechanisms, with measured average bandwidth overhead of about 31\% and latency increase of about 34\%\cite{wang2017walkie}. Huang et al. (2025) proposed asymmetric defense (STAP) reducing attack accuracy to 48.3\% with only 18\% bandwidth overhead\cite{huang2025stap}. Wright et al. (2009) formalized traffic shaping as a convex optimization problem: given source distribution \(X\) and target distribution \(Y\), find transformation matrix \(A \geq 0\) such that \(AX = Y\) (where \(\sum_i A_{ij} = 1\) ensures each column is a probability distribution), while minimizing expected bandwidth overhead \(\sum_{i,j} x_j a_{ij} |s_i - s_j|\), where \(s_i, s_j\) are packet sizes\cite{wright2009traffic}. These works revealed the theoretical costs of defense but only targeted the specific scenario of website fingerprinting, did not establish general efficiency-leakage relationships, and did not explain why side channels necessarily exist in the absence of defense.

Quantification of mutual information leakage provided information-theoretic tools for side-channel analysis. Li et al. (2018) first systematically applied mutual information \(I(X;Y)\) to quantify website fingerprinting leakage at ACM CCS\cite{li2018measuring}. In experiments on 100 websites in a closed Tor network environment, they found that even though Tor uses encryption and obfuscation techniques, traffic features still leak substantial website information. Information leakage \(I(F;W)\) from individual features reaches up to 3.45 bits (from the rounded outbound packet count feature), 54.55\% of features leak less than 1 bit, while combined multiple features achieve total information leakage of approximately 6.6 bits, approaching the theoretical limit \(\log_2 100 \approx 6.64\) bits. Cherubin (2017) proposed \((\xi, \Phi)\)-privacy metrics based on Bayes error lower bounds, defining defense security from an information-theoretic perspective\cite{cherubin2017bayes}. These works revealed the phenomenon of information leakage persisting after encryption but did not explain its inevitability from a system design perspective—namely, why \(I(X;Y)\) is necessarily greater than zero in encryption systems satisfying practical conditions?

Unlike this direct approach of measuring leakage using mutual information or Bayes error lower bounds, Fu et al. and subsequent works\cite{fu2021realtime,fu2024detecting,fu2023detecting} model encrypted traffic features as random variables or signals, performing differential entropy, information loss, KL divergence, and separability/robustness analysis in frequency-domain spectral features, length patterns, or flow interaction graph spaces, focusing on evaluating the effectiveness of given features and detection methods in specific tasks such as malicious traffic detection and tunnel traffic identification. In contrast, this paper adopts a system-level perspective of composite channel–mutual information lower bounds, discussing whether encrypted traffic side-channel leakage is ``inevitable'' and its theoretical lower bound without pre-fixing feature forms.

Formal security analysis of anonymous systems attempts to prove privacy guarantees in more rigorous frameworks. Camenisch and Lysyanskaya (2005) formalized onion routing in the Universal Composability (UC) framework, providing composable security definitions\cite{camenisch2005formal}. Feigenbaum et al. (2007) modeled onion routing using probabilistic I/O automata, providing formal anonymity analysis against active timing attacks\cite{feigenbaum2007model}. Danezis and Goldberg's (2009) Sphinx model provides compact message formats for mix networks, proving unlinkability and path length hiding under the random oracle model\cite{danezis2009sphinx}. These formal methods provide rigorous security proofs under specific threat models but mainly target security analysis of specific anonymity protocols, without establishing a general theoretical framework from an information-theoretic perspective for the inevitability of side-channel leakage in encrypted communications.

Despite significant progress in respective fields, obvious gaps remain in the theoretical foundation of side-channel existence: (1) Lack of formal causal framework—how to rigorously model the information propagation process from application semantics to observable features? How to mathematically characterize the layer-by-layer effects of encryption, encapsulation, and transmission? (2) Intrinsic connection between efficiency constraints and leakage unclear—how does efficiency prioritization in system design necessarily lead to side channels? Under what verifiable conditions is leakage inevitable? (3) Lack of computable leakage boundaries—given system parameters, what is the lower bound of side-channel leakage? How to predict actual attack performance from theoretical lower bounds? This paper systematically fills these theoretical gaps by constructing formal models, proving existence theorems, and establishing explicit lower bounds, providing a rigorous mathematical foundation for understanding the fundamental causes of side channels.

\section{Formal Model for Side-Channel Analysis}

\subsection{Basic Definitions and Modeling Conventions}

The core approach of this paper is to abstract the entire process of ``generation, encapsulation, encryption, transmission, observation'' as a composite channel consisting of measurable mappings and random channels, making the mutual information between semantic variables and observable features strictly definable and characterizable by the data processing inequality. We first introduce the basic definitions and conventions of the modeling.

\textbf{Time and Randomness:} Time is modeled as the continuous non-negative axis \(\mathbb{T}=\mathbb{R}_{\ge 0}\). All random variables (semantic \(X\), application-side randomness \(U_A\), protocol-side \(U_\Pi\), encryption-side \(U_\Phi\), network-side \(U_N\), observation-side \(U_\Theta\)) are defined on a common probability space \((\Omega_0,\mathcal{F}_0,\mathbb{P})\); the specific construction of this space does not affect subsequent analysis and is used to ensure the well-definedness of joint distributions and conditional expectations.

\textbf{Point Processes and Sequence Representations:} Point processes \(\Xi_A,\Xi_P,\Xi_C,\Xi_N\) (for intuition, subsequently referred to as ``message sequences, plaintext packet sequences, ciphertext packet sequences, arrival packet sequences'') and observed features \(Y\) have sample paths taking values in corresponding measurable spaces.

\textbf{Unified Trajectory Space and Metric:} To avoid inconsistent domains for cross-layer statistics, we introduce a unified marked point process trajectory space \(\mathcal{Z}\) (e.g., using finite marked counting measure space, or embedding window-within trajectories into Skorokhod space). Let \(e_P:\mathrm{range}(\Xi_P)\to\mathcal{Z}\), \(e_N:\mathrm{range}(\Xi_N)\to\mathcal{Z}\) be measurable embeddings from each layer to \(\mathcal{Z}\). All window-level statistics below are viewed as measurable functions \(\varphi:\mathcal{Z}\to[-M,M]\). The metric \(d\) is defined on \(\mathcal{Z}\), simultaneously measuring timing and marking (length, direction) differences, and is compatible with the Lipschitz conditions below. To avoid unnecessary technical burden, we assume \(\mathcal{X}\), \(\mathcal{Y}\), \(\mathrm{range}(\Xi_A)\), \(\mathrm{range}(\Xi_P)\), \(\mathrm{range}(\Xi_C)\), \(\mathrm{range}(\Xi_N)\), and \(\mathcal{Z}\) are all standard Borel spaces.

To clearly distinguish ``system design'' from ``external observation'', we provide the following definition of side-channel analysis model:

\begin{definition}[Side-Channel Analysis Model]\label{def:sigma} 
The side-channel analysis model is denoted as \(\Sigma=(\Gamma,\Omega)\), where the encrypted communication model is denoted as \(\Gamma=(A,\Pi,\Phi,N)\), consisting of four causal operators: application generation \(A\), protocol encapsulation \(\Pi\), encryption transformation \(\Phi\), network transmission \(N\); \(\Omega\) is the observation model, characterizing the side-channel analyst's passive monitoring and feature extraction capabilities at various network layers and positions.
\end{definition}

Under this definition, system \(\Gamma\) only determines ``how to generate and transmit ciphertext'', and observation model \(\Omega\) only determines ``what the analyst can see''. The leakage amount \(L(\Gamma,\Omega)\) is the result of the composite action of both. This separation allows the existence theorem below to hold under the weakest system and observation assumptions.

\subsection{Encrypted Communication Model}

The encrypted communication model \(\Gamma=(A,\Pi,\Phi,N)\) consists of four causal operators that act on message point processes and packet sequences on the time axis. We characterize these transformations as causally measurable mappings, only requiring them to maintain temporal causality and measurability, without making strong assumptions about the specific forms of probability distributions (such as independence or stationarity).

\textbf{(1) Application Layer (Operator \(A\))}

Let \(\mathcal{X}\) be the application semantic space (such as website ID, application category, video content, etc.), and let \(X\in\mathcal{X}\) be the semantic variable. The application generates a message sequence on the time axis:
\begin{equation}
\Xi_A=\{(\tau_k,m_k)\}_{k\ge 1},\quad \tau_k\in\mathbb{T},\ m_k\in\mathcal{M},
\end{equation}
which is jointly determined by \(X\) and exogenous noise, and can be written as a causally measurable mapping \(\Xi_A=\mathcal{G}_A(X,U_A)\). Here \(\mathcal{M}\) is the message set, and \(U_A\) is application-side randomness (user behavior, business logic jitter, etc.). We do not require stationarity or independence, only that \(\mathcal{G}_A\) is causal and measurable.

\textbf{(2) Protocol Layer (Operator \(\Pi\))}

The protocol stack maps the message sequence \(\Xi_A\) to a segmented and encapsulated plaintext packet sequence:
\begin{equation}
\Xi_P=\{(t_i,\ell_i,\mathrm{dir}_i,h_i,b_i)\}_{i\ge 1},
\end{equation}
where \(t_i\in\mathbb{T}\) is the sending time, \(\ell_i=|h_i|+|b_i|\) is the length, \(\mathrm{dir}_i\in\{\uparrow,\downarrow\}\) is the direction, \(h_i\) is the concatenation of headers from various layers, and \(b_i\) is the plaintext payload. Expressed as a causal mapping \(\Xi_P=\Pi(\Xi_A,U_\Pi)\), where \(U_\Pi\) represents protocol-side randomness (Nagle, segmentation boundaries, congestion control adaptation, etc.). \(\Pi\) maintains causality: \(t_i\) depends only on \(\{(\tau_k,m_k):\tau_k\le t_i\}\) and past protocol state.

\textbf{(3) Encryption Layer (Operator \(\Phi\))}

Encryption maps the plaintext packet sequence \(\Xi_P\) to a ciphertext packet sequence:
\begin{equation}
\Xi_C=\{(t_i',\ell_i',\mathrm{dir}_i,h_i',c_i)\}_{i\ge 1}=\Phi(\Xi_P,U_\Phi),
\end{equation}
where \(h_i'\) are visible or semi-visible header fields, and \(c_i\) is the ciphertext payload. We only use two general properties: (i) \emph{Semantic independence}: under the ideal cipher assumption, given input length and public parameters, \(c_i\) is conditionally independent of plaintext content; (ii) \emph{Determinism of length transformation}: there exists a deterministic function \(g_{\mathrm{len}}\) such that \(\ell_i'=g_{\mathrm{len}}(\ell_i;\theta_\Phi)\), where \(\theta_\Phi\) contains block size, record overhead, optional padding, etc. Real-world protocols (TLS 1.3, QUIC) satisfy this intuitive property: content is scrambled, but length and timing are only affected by minor alignment and record formatting. Time mapping maintains causality: \(t_i'\ge t_i\), and \(t_i'-t_i\) is determined by implementation and scheduling.

\textbf{(4) Network Layer (Operator \(N\))}

The network transmits the sender-side ciphertext sequence \(\Xi_C\) as an arrival packet sequence on the observation path:
\begin{equation}
\Xi_N=\{(\tilde t_j,\tilde \ell_j,\mathrm{dir}_j,\tilde h_j,\tilde c_j)\}_{j\ge 1}=N(\Xi_C,U_N),
\end{equation}
where \(U_N\) characterizes uncertainties such as queuing, routing, packet loss, retransmission, reordering, framing, and multiplexing. \(N\) is viewed as a causal channel of a random time-varying queuing network; independence or Markov assumptions are not required.

A fundamental assumption about \(\Gamma\) is that it produces statistically distinguishable traffic patterns at the protocol layer. To this end, we provide the following definition of semantic distinguishability.

\begin{definition}[Semantic Distinguishability]\label{def:semantic-distinguishability}
Given window \(T>0\), the system \(\Gamma\) is said to have \(\bar\Delta\)-distinguishability in the \(\mathcal{Z}\) representation induced at the protocol layer if there exist a semantic pair \(x\neq x'\in\mathcal{X}\) and a bounded measurable statistic \(\varphi:\mathcal{Z}\to[-M,M]\) such that
\begin{equation}
\big|\mathbb{E}[\varphi(e_P(\Xi_P|_{[0,T]}))\mid X=x]-\mathbb{E}[\varphi(e_P(\Xi_P|_{[0,T]}))\mid X=x']\big|\ \ge\ \bar\Delta.
\end{equation}
The system \(\Gamma\) is said to be distinguishable at the protocol layer if there exists \(\bar\Delta>0\) such that it has \(\bar\Delta\)-distinguishability.
\end{definition}

This definition characterizes an intrinsic property of the system: different semantics, after application generation and protocol encapsulation, produce sufficiently large differences in the conditional expectations of some statistic in the unified trajectory space \(\mathcal{Z}\). The existence of \(\bar\Delta\) excludes the degenerate case of ``arbitrarily small expectation differences'', ensuring that distinguishability is operationally meaningful in a statistical sense. Common statistics \(\varphi\) include: total bytes in window, ratio of upstream to downstream packets, weighted average of packet intervals, etc. Different application semantics often exhibit differences in these statistics: high byte counts for video streams, balanced upstream/downstream ratios for instant messaging, short-term high-density patterns for web browsing—these patterns maintain statistical separability even after protocol encapsulation.

Real-world encrypted communication system design follows the efficiency-first principle: maximizing resource utilization efficiency to support upper-layer business needs while ensuring cryptographic security. Business has quantifiable requirements for multiple performance dimensions:
\begin{itemize}[topsep=3pt,itemsep=2pt]
\item Bandwidth overhead: additional bytes introduced by the protocol should be controlled within an acceptable range (e.g., TLS 1.3 approximately 5\%)
\item End-to-end latency: delays introduced by encryption, framing, and scheduling should be below application perception thresholds (e.g., web pages \(<\)200ms, VoIP \(<\)150ms)
\item Throughput: effective throughput should not significantly decrease due to padding or artificial delays
\item Protocol compatibility: must be compatible with existing network infrastructure (MTU, congestion control, NAT traversal, etc.)
\end{itemize}

The direct consequence of efficiency-first design is: the mapping from plaintext protocol layer \(\Xi_P\) to arrival packet sequence \(\Xi_N\) necessarily maintains non-degeneracy.

\begin{definition}[Mapping Non-Degeneracy]\label{def:non-degeneracy}
On a fixed window \([0,T]\), the encrypted communication system \(\Gamma\) is said to have window-level mapping non-degeneracy if there exists a constant \(C_T<\infty\) such that for all \(x\in\mathcal{X}\),
\begin{equation}
\mathbb{E}\Big[d\big(e_P(\Xi_P|_{[0,T]}),\,e_N(\Xi_N|_{[0,T]})\big)\,\Big|\,X=x\Big]\ \le\ C_T.
\end{equation}
Intuitively, this requires that encryption and transmission mappings do not excessively distort the statistical structure of the plaintext packet sequence—the joint distribution of marks (length, direction) and timing of the point process remains ``close'' before and after mapping.
\end{definition}

Non-degeneracy directly stems from efficiency constraints, such as excessive padding causing the metric to diverge in the length component (bandwidth overhead exceeding limits), excessive delays or jitter causing the metric to diverge in the timing component (end-to-end latency exceeding limits, breaking real-time properties), or completely scrambling the marks and timing structure of the packet sequence causing the metric to diverge (protocol semantics broken, business unusable). ``Window-level'' means the metric acts on point process segments within a finite time window \([0,T]\) (e.g., 10 seconds or one session), allowing single-packet perturbations (e.g., TCP retransmissions, reordering) but constraining cumulative deviation. The specific form of metric \(d\) can be varied, as long as it can capture the joint statistical structure of length, timing, and direction of point processes. The conclusions of this paper do not depend on specific metric choices; the establishment of subsequent theorems relies only on abstract non-degeneracy—namely, the existence of some suitable metric \(d\) such that \(\mathbb{E}[d]\le C_T<\infty\). In typical network scenarios, a commonly used choice is to define \(d\) as a weighted sum of ``length deviation + packet count deviation + latency deviation'', or implicit distances in some high-dimensional space of flow features characterized by machine learning and deep learning. In this case, \(C_T\) can be understood as the maximum acceptable deviation jointly defined by bandwidth, packet count overhead, and end-to-end latency jitter in that window.

The above semantic distinguishability and mapping non-degeneracy jointly constitute intrinsic properties of encrypted communication systems. They characterize ``the statistical variability that the system necessarily retains under the premise of business availability''. The next section will discuss how side-channel analysts capture these properties through observation model \(\Omega\), ultimately deriving the existence of side-channel leakage.

\subsection{Observation Model}

The observation model \(\Omega\) characterizes the side-channel analyst's passive access capabilities to various network layers and positions, as well as the feature extraction process from arrival packet sequences. Although observers cannot directly access the plaintext protocol layer \(\Xi_P\), they can infer semantic information by observing arrival packet sequences \(\Xi_N\) and extracting features. This section defines the composition of the observation model, the generation of observed features, and non-degeneracy conditions ensuring that the observation channel does not degenerate to a constant.

\begin{definition}[Observation Model]\label{def:observer}
The side-channel analysis model \(\Omega=(\mathcal{L}_{\mathrm{acc}},\mathcal{O}_{\mathrm{acc}},\Theta)\) consists of three components:
\begin{itemize}[topsep=3pt,itemsep=2pt]
\item \(\mathcal{L}_{\mathrm{acc}}\): Accessible layer set (such as IP layer, UDP/TCP layer, QUIC/TLS record layer, etc.)
\item \(\mathcal{O}_{\mathrm{acc}}\): Observable position set (links, switch ports, host network interfaces, etc.)
\item \(\Theta\): Causally measurable feature extraction operator
\end{itemize}
Given arrival packet sequence \(\Xi_N\), the observer obtains feature sequence through \(\Theta\):
\begin{equation}
Y=\Theta(\Xi_N;\mathcal{L}_{\mathrm{acc}},\mathcal{O}_{\mathrm{acc}},U_\Theta),
\end{equation}
where \(U_\Theta\) represents observation-side uncertainty (timestamp precision, sampling strategy, counting granularity, etc.).
\end{definition}

The feature extraction operator \(\Theta\) typically includes the following operations: extracting packet length sequences, computing packet arrival time intervals, counting upstream/downstream packet ratios, constructing burst pattern descriptors, computing window-level statistics (total bytes, packet counts, rates, etc.). The specific form of \(\Theta\) depends on the observer's technical capabilities and analysis objectives but must maintain causality: observed features at time \(t\) depend only on arrival packets at or before time \(t\).

Connecting system model \(\Gamma\) and observation model \(\Omega\) end-to-end yields the complete causal chain from semantics to observed features, as shown in Figure 1:
\begin{equation}
X \xrightarrow{\ \mathcal{G}_A\ } \Xi_A \xrightarrow{\ \Pi\ } \Xi_P \xrightarrow{\ \Phi\ } \Xi_C \xrightarrow{\ N\ } \Xi_N \xrightarrow{\ \Theta\ } Y.
\end{equation}
where \(X\) is the sensitive semantic variable, \(\Xi_A,\Xi_P,\Xi_C,\Xi_N\) are message sequence, plaintext packet sequence, ciphertext packet sequence, arrival packet sequence (all point processes), and \(Y\) is the observed feature sequence. This causal chain induces two key properties:

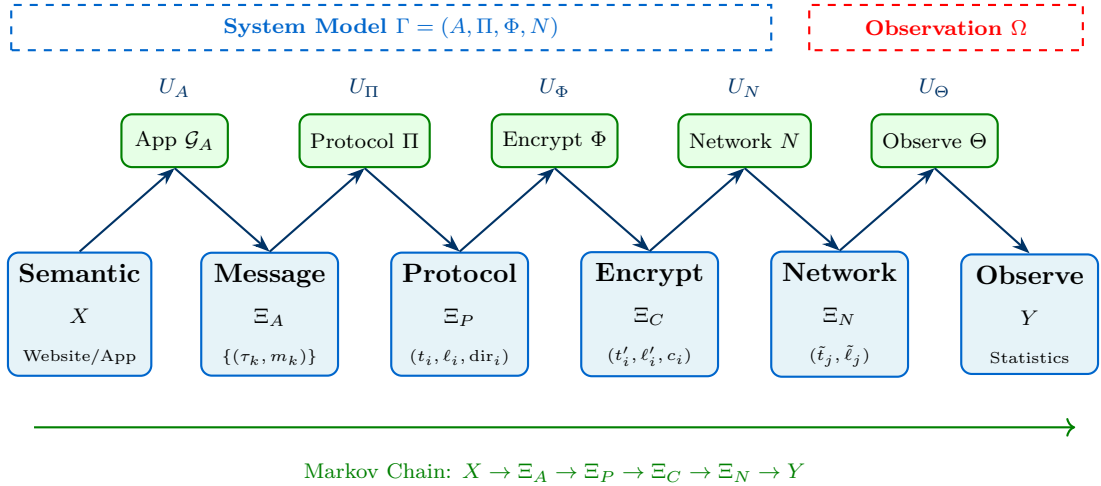
\begin{figure}[htbp]
\centering
\begin{tikzpicture}[
    layer/.style={rectangle, rounded corners, draw=copenblue, thick, 
                  minimum height=1.5cm, minimum width=1.8cm, align=center, fill=copenlight!30},
    process/.style={rectangle, rounded corners, draw=copengreen, thick,
                    minimum height=0.7cm, minimum width=1.4cm, align=center, fill=green!10},
    arrow/.style={-Stealth, thick, copendark}
]
\node[layer] (X) at (0,0) {\footnotesize\textbf{Semantic}\\[1pt]\scriptsize \(X\)\\[0pt]{\tiny Website/App}};

\node[layer] (XiA) at (2.5,0) {\footnotesize\textbf{Message}\\[1pt]\scriptsize \(\Xi_A\)\\[0pt]{\tiny \(\{(\tau_k,m_k)\}\)}};

\node[layer] (XiP) at (5,0) {\footnotesize\textbf{Protocol}\\[1pt]\scriptsize \(\Xi_P\)\\[0pt]{\tiny \((t_i,\ell_i,\mathrm{dir}_i)\)}};

\node[layer] (XiC) at (7.5,0) {\footnotesize\textbf{Encrypt}\\[1pt]\scriptsize \(\Xi_C\)\\[0pt]{\tiny \((t_i',\ell_i',c_i)\)}};

\node[layer] (XiN) at (10,0) {\footnotesize\textbf{Network}\\[1pt]\scriptsize \(\Xi_N\)\\[0pt]{\tiny \((\tilde{t}_j,\tilde{\ell}_j)\)}};

\node[layer] (Y) at (12.5,0) {\footnotesize\textbf{Observe}\\[1pt]\scriptsize \(Y\)\\[0pt]{\tiny Statistics}};

\node[process] (A) at (1.25,2.3) {\scriptsize App \(\mathcal{G}_A\)};
\node[process] (Pi) at (3.75,2.3) {\scriptsize Protocol \(\Pi\)};
\node[process] (Phi) at (6.25,2.3) {\scriptsize Encrypt \(\Phi\)};
\node[process] (N) at (8.75,2.3) {\scriptsize Network \(N\)};
\node[process] (Theta) at (11.25,2.3) {\scriptsize Observe \(\Theta\)};

\node[font=\scriptsize, copendark] at (1.25,3) {\(U_A\)};
\node[font=\scriptsize, copendark] at (3.75,3) {\(U_\Pi\)};
\node[font=\scriptsize, copendark] at (6.25,3) {\(U_\Phi\)};
\node[font=\scriptsize, copendark] at (8.75,3) {\(U_N\)};
\node[font=\scriptsize, copendark] at (11.25,3) {\(U_\Theta\)};

\draw[arrow] (X.north) -- (A.south);
\draw[arrow] (XiA.north) -- (Pi.south);
\draw[arrow] (XiP.north) -- (Phi.south);
\draw[arrow] (XiC.north) -- (N.south);
\draw[arrow] (XiN.north) -- (Theta.south);

\draw[arrow] (A.south) -- (XiA.north);
\draw[arrow] (Pi.south) -- (XiP.north);
\draw[arrow] (Phi.south) -- (XiC.north);
\draw[arrow] (N.south) -- (XiN.north);
\draw[arrow] (Theta.south) -- (Y.north);

\draw[copenblue, dashed, thick] (-0.9,3.5) rectangle (9.1,4.1);
\node[font=\scriptsize\bfseries, copenblue] at (4.1,3.8) {System Model \(\Gamma=(A,\Pi,\Phi,N)\)};

\draw[red, dashed, thick] (9.6,3.5) rectangle (13.2,4.1);
\node[font=\scriptsize\bfseries, red] at (11.4,3.8) {Observation \(\Omega\)};

\draw[copengreen, thick, ->] (-0.6,-1.5) -- (13.1,-1.5);
\node[font=\scriptsize, copengreen, align=center] at (6.25,-2.1) 
    {Markov Chain: \(X \to \Xi_A \to \Xi_P \to \Xi_C \to \Xi_N \to Y\)};
\end{tikzpicture}
\caption{Complete causal chain from semantics to observed features}
\label{fig:causal_chain}
\end{figure}

\begin{proposition}[Existence and Measurability of Observation Channel]\label{prop:kernel}
Under the conditions that the above spaces are standard Borel spaces and \(N\), \(\Theta\) are random kernels (or compositions of measurable mappings and random kernels), there exists a random kernel \(K_\Sigma(\mathrm{d}y\mid x)\) from \(\mathcal{X}\) to \(\mathcal{Y}\) such that \(Y\mid X=x \sim K_\Sigma(\cdot\mid x)\), hence \(I(X;Y)\) is well-defined. This conclusion stems from the closure of random kernels under composition and the existence of regular conditional probability.
\end{proposition}

\begin{proof}
On standard Borel spaces, compositions of measurable mappings and random kernels remain random kernels, and regular conditional probability exists\cite{gray2011probability}. Since \(\mathcal{G}_A,\Pi,\Phi,\Theta\) are causally measurable mappings and \(N\) is a causal random channel, their composition yields \(K_\Sigma(\mathrm{d}y\mid x)\).
\end{proof}

\begin{proposition}[Data Processing Structure]\label{prop:dpi}
The composite chain induces the Markov relationship
\begin{equation}
X \to \Xi_A \to \Xi_P \to \Xi_C \to \Xi_N \to Y,
\end{equation}
thus for any intermediate layer variable \(Z\), \(I(X;Y)\le I(X;Z)\). In particular, letting \(Y_{\mathrm{raw}}=\Xi_P\) or \(Y_{\mathrm{raw}}=\Xi_C\), we have \(I(X;Y)\le I(X;Y_{\mathrm{raw}})\).
\end{proposition}

\begin{proof}
By causality and measurability, each arrow in the chain is a composition of random kernels, satisfying the Markov property. By the data processing inequality\cite{cover2006elements}, for the Markov chain \(X \to \Xi_A \to \Xi_P \to \Xi_C \to \Xi_N \to Y\), we have \(I(X;Y) \le I(X;Z)\) for any intermediate variable \(Z\).
\end{proof}

A key property of observation model \(\Omega\) is whether it can preserve the statistical variability of arrival packet sequences. If the observation mapping \(\Theta\) compresses all inputs to a constant (e.g., only recording ``traffic exists'' while discarding all length, timing, direction information), then regardless of how rich the statistical patterns produced by system \(\Gamma\), observed features \(Y\) cannot reflect semantic differences. To exclude this degenerate case, we provide the following definition:

\begin{definition}[Non-Degenerate Observation (with respect to statistic used)]\label{def:nondegenerate}
Given window \(T\) and statistic \(\varphi\) from Definition \ref{def:semantic-distinguishability}, the observation model \(\Omega\) is said to be non-degenerate with respect to this statistic if there exist a constant \(\rho\in(0,1]\) and a bounded measurable mapping \(\psi:\mathcal{Y}\to[-1,1]\) (depending only on \(T\) and \(\varphi\), independent of specific semantic pairs and their priors) such that for any semantic pair \(x\neq x'\in\mathcal{X}\) and any \(\delta>0\),
\begin{equation}
\begin{aligned}
&\Big|\mathbb{E}\big[\varphi(e_N(\Xi_N|_{[0,T]}))\mid X=x\big]-\mathbb{E}\big[\varphi(e_N(\Xi_N|_{[0,T]}))\mid X=x'\big]\Big|\ \ge\ \delta \\
&\qquad\Rightarrow\quad
\Big|\mathbb{E}[\psi(Y)\mid X=x]-\mathbb{E}[\psi(Y)\mid X=x']\Big|\ \ge\ \rho\,\delta.
\end{aligned}
\end{equation}
Intuitively, this requires that the observation mapping preserves at least a positive proportion of the conditional expectation difference of statistic \(\varphi\)—even if there is information loss (\(\rho<1\)), statistical differences are not completely erased (\(\rho>0\)). The constant \(\rho\) and mapping \(\psi\) depend only on the observation model \(\Omega\), window \(T\), and the structure of statistic \(\varphi\), not on the specific semantic pairs being compared or their prior distributions, reflecting an inherent property of the observation model.
\end{definition}

This definition excludes extremely degenerate observation models, such as: binary indicators recording only ``whether there is traffic'', aggregators counting only total session duration while discarding all fine-grained information, constant mappings returning fixed feature vectors for all traffic, etc. Real-world side-channel analysis models typically satisfy non-degeneracy because observers wish to maximize extractable information and will retain key dimensions such as packet length, timing, and direction.

Thus far, we have constructed the complete side-channel analysis model: system \(\Gamma\) produces distinguishable statistical patterns at the protocol layer, which, after encryption and transmission satisfying mapping non-degeneracy, arrive at the observation path; the non-degenerate observation model \(\Omega\) captures these patterns and extracts features \(Y\). The next section will prove, based on this framework, that under the premise that the system satisfies window-level mapping non-degeneracy (Definition \ref{def:non-degeneracy}), the mutual information \(I(X;Y)\) between observed features and semantics is necessarily strictly greater than zero, making side-channel leakage inevitable.

\section{Side-Channel Leakage Existence Theorem}

Based on the modeling framework of the previous section, this section provides a rigorous statement and proof of the inevitability of side-channel leakage. We first establish a stable propagation chain of expectation differences for binary semantic pairs and provide an explicit lower bound on mutual information, then generalize to the general case of multi-semantic spaces. The entire argument relies only on the causally measurable composite channel structure, the data processing inequality, and the stable propagation of bounded metrics, without depending on specific protocol details or particular metric choices.

\subsection{Technical Strengthening Assumptions}

Let \(X\) be the application semantic variable and \(Y\) be the observable features obtained by the observer under model \(\Sigma=(\Gamma,\Omega)\). The previous section provided the Markov chain \(X\to \Xi_A\to \Xi_P\to \Xi_C\to \Xi_N\to Y\) and corresponding random kernel \(K_\Sigma(\mathrm{d}y\mid x)\) (Proposition \ref{prop:kernel}).

Definition \ref{def:semantic-distinguishability} provided the basic concept of semantic distinguishability: there exists a statistic \(\varphi\) on the unified trajectory space \(\mathcal{Z}\) induced at the protocol layer such that conditional expectation differences for different semantics are \(\ge\bar\Delta\). To ensure this distinguishability can stably propagate to the observation layer, we need to impose Lipschitz continuity constraints on statistic \(\varphi\):

\begin{assumption}[Lipschitz Robustness]\label{ass:lipschitz}
The \(\varphi\) in Definition \ref{def:semantic-distinguishability} is \(L_\varphi\)-Lipschitz continuous with respect to metric \(d\) on \(\mathcal{Z}\):
\begin{equation}
|\varphi(z)-\varphi(z')|\ \le\ L_\varphi\, d(z,z'),\quad \forall z,z'\in\mathcal{Z}.
\end{equation}
\end{assumption}

This assumption ensures the stability of statistics under trajectory perturbations—small changes in trajectories lead only to small changes in statistics. Common bounded Lipschitz statistics include: truncated versions of total bytes in window, ratios of upstream/downstream packets (naturally bounded), saturated versions of packet arrival intervals, etc. In practice, the vast majority of statistical features used for traffic analysis can be converted into Lipschitz continuous functions through appropriate truncation or normalization.

\subsection{Side-Channel Leakage Theorem for Binary Semantic Pairs}

We first establish a strict leakage lower bound for binary semantic pairs. Denote the leakage amount as \(L(\Gamma,\Omega)=I(X;Y)\).

\begin{theorem}[Binary Semantic Side-Channel Leakage Theorem]\label{thm:binary-leakage}
Under the conditions of Proposition \ref{prop:kernel}, fix window \(T>0\). Suppose there exists a distinguishable semantic pair \(x\neq x'\in\mathcal{X}\) satisfying the prior positive mass condition \(\mathbb{P}(X=x)>0\), \(\mathbb{P}(X=x')>0\). If the system satisfies the following conditions:

\begin{enumerate}[label=(\roman*)]
\item \textbf{Mapping Non-Degeneracy} (instantiation of Definition \ref{def:non-degeneracy} for the semantic pair): there exist metric \(d\) and constant \(C<\infty\) such that for this semantic pair,
\begin{equation}
\begin{aligned}
\max\Big\{\mathbb{E}\Big[d\big(e_P(\Xi_P|_{[0,T]}),\,e_N(\Xi_N|_{[0,T]})\big)\,\Big|\,X=x\Big],\,\\
\mathbb{E}\Big[d\big(e_P(\Xi_P|_{[0,T]}),\,e_N(\Xi_N|_{[0,T]})\big)\,\Big|\,X=x'\Big]\Big\}\ \le\ C.
\end{aligned}
\end{equation}

\item \textbf{Semantic Distinguishability} (Definition \ref{def:semantic-distinguishability}): there exist \(\bar\Delta>0\) and bounded measurable statistic \(\varphi:\mathcal{Z}\to[-M,M]\) such that
\begin{equation}
\big|\mathbb{E}[\varphi(e_P(\Xi_P|_{[0,T]}))\mid X=x]-\mathbb{E}[\varphi(e_P(\Xi_P|_{[0,T]}))\mid X=x']\big|\ \ge\ \bar\Delta.
\end{equation}

\item \textbf{Lipschitz Robustness} (Assumption \ref{ass:lipschitz}): the above statistic \(\varphi\) is \(L_\varphi\)-Lipschitz continuous with respect to metric \(d\) on \(\mathcal{Z}\):
\begin{equation}
|\varphi(z)-\varphi(z')|\ \le\ L_\varphi\, d(z,z'),\quad \forall z,z'\in\mathcal{Z}.
\end{equation}

\item \textbf{Non-Degenerate Observation} (Definition \ref{def:nondegenerate} for statistic \(\varphi\)): there exist constant \(\rho\in(0,1]\) and bounded measurable \(\psi:\mathcal{Y}\to[-1,1]\) (depending only on \(T\) and \(\varphi\), independent of semantic pairs) such that for any semantic pair \(a\neq b\in\mathcal{X}\),
\begin{equation}
\begin{aligned}
&\Big|\mathbb{E}\big[\varphi(e_N(\Xi_N|_{[0,T]}))\mid X=a\big]-\mathbb{E}\big[\varphi(e_N(\Xi_N|_{[0,T]}))\mid X=b\big]\Big|\ \ge\ \delta \\
&\qquad\Rightarrow\quad
\Big|\mathbb{E}[\psi(Y)\mid X=a]-\mathbb{E}[\psi(Y)\mid X=b]\Big|\ \ge\ \rho\,\delta.
\end{aligned}
\end{equation}

\item \textbf{Distinguishability Propagation Condition}: the metric deviation bound \(C\), Lipschitz constant \(L_\varphi\), and distinguishability margin \(\bar\Delta\) satisfy
\begin{equation}
C < \frac{\bar\Delta}{2L_\varphi}.
\end{equation}
\end{enumerate}

Then the mutual information between observed features and semantics satisfies \(I(X;Y) > 0\), with explicit lower bound
\begin{equation}
I(X;Y)\ \ge\ \frac{2}{\ln 2}\,\mathbb{P}(X=x)\mathbb{P}(X=x')\,
\left(\frac{\rho\,[\bar\Delta-2L_\varphi C]}{2}\right)^2.
\end{equation}

In particular, restricting to the binary equal-prior subproblem (i.e., \(X\in\{x,x'\}\) and \(\mathbb{P}(X=x)=\mathbb{P}(X=x')=1/2\)), the lower bound becomes
\begin{equation}
I(X;Y)\ \ge\ \frac{1}{2\ln 2}\left(\frac{\rho\,[\bar\Delta-2L_\varphi C]}{2}\right)^2.
\end{equation}
\end{theorem}

The theorem focuses on analysis of binary semantic pairs \((x,x')\), which is an essential requirement of the proof technique—arguments based on Pinsker's inequality and total variation are naturally pairwise. Condition (i) only requires this semantic pair to satisfy mapping non-degeneracy; in efficiency-prioritized practical systems, if the mapping maintains non-degeneracy for all semantics, it naturally holds for any semantic pair. Condition (v) characterizes the boundary of side-channel existence: the metric deviation \(C\) cannot be too large, otherwise the distinguishability at the protocol layer will be completely erased by perturbations during propagation to the network layer. In typical website fingerprinting scenarios, \(x\) and \(x'\) can be understood as two semantic classes ``visiting website A'' and ``visiting website B'', where \(\bar\Delta\) can be chosen as expectation differences in aggregate statistics such as total bytes, upstream/downstream packet ratios, burst duration/payload within a fixed time window; \(C\) corresponds to the average deviation in length and arrival time in the unified trajectory space \(\mathcal{Z}\) due to encryption encapsulation, TCP retransmissions, and random jitter within the same window; and \(\rho\) reflects the proportion of these differences that the observation link can still preserve under limitations such as sampling granularity and timestamp precision.

The proof requires the following lemma:

\begin{lemma}[Relationship between Expectation Difference and Total Variation]\label{lem:exp-to-tv}
Let \(f:\mathcal{S}\to[-M,M]\) be a bounded measurable function, and \(P,Q\) be two probability measures on \(\mathcal{S}\). If
\begin{equation}
|\mathbb{E}_P[f] - \mathbb{E}_Q[f]| \ge \delta,
\end{equation}
then the total variation distance satisfies
\begin{equation}
\mathrm{TV}(P,Q) \ge \frac{\delta}{2M}.
\end{equation}
\end{lemma}

\begin{proof}
By the dual representation of total variation,
\begin{equation}
\mathrm{TV}(P,Q) = \frac{1}{2}\sup_{\|g\|_\infty\le 1} |\mathbb{E}_P[g] - \mathbb{E}_Q[g]|.
\end{equation}
For bounded function \(f:\mathcal{S}\to[-M,M]\), normalize \(g := f/M\) to get \(\|g\|_\infty = 1\), therefore
\begin{equation}
|\mathbb{E}_P[f] - \mathbb{E}_Q[f]| = M|\mathbb{E}_P[g] - \mathbb{E}_Q[g]| 
\le M \cdot 2\cdot\mathrm{TV}(P,Q) = 2M\cdot\mathrm{TV}(P,Q).
\end{equation}
If \(|\mathbb{E}_P[f] - \mathbb{E}_Q[f]| \ge \delta\), then \(\mathrm{TV}(P,Q) \ge \frac{\delta}{2M}\).
\end{proof}

\begin{proof}[Proof of Theorem \ref{thm:binary-leakage}]
The proof unfolds along ``\(\bar\Delta\to(L_\varphi)\to d\to(C_T)\to\rho\to I(X;Y)\)'': using Lipschitz property to bound protocol-layer differences to trajectory space, then absorbing perturbations using \(C_T\) and characterizing observable proportion with \(\rho\), finally concatenating inequalities to obtain the mutual information lower bound. The proof consists of four steps: establishing expectation difference at protocol layer \(\Xi_P\), propagating to network layer \(\Xi_N\), propagating to observation layer \(Y\), and finally converting to mutual information. The entire derivation proceeds in the unified trajectory space \(\mathcal{Z}\), relying on the distinguishability propagation chain formed by conditions (i)-(v).

\textbf{Step 1: Expectation difference at protocol layer is directly given by condition (ii).}

By condition (ii) semantic distinguishability, there exists bounded statistic \(\varphi:\mathcal{Z}\to[-M,M]\) satisfying
\begin{equation}
\big|\mathbb{E}[\varphi(e_P(\Xi_P|_{[0,T]}))\mid X=x] - \mathbb{E}[\varphi(e_P(\Xi_P|_{[0,T]}))\mid X=x']\big| \ge \bar\Delta.
\end{equation}

\textbf{Step 2: Propagation of expectation difference from protocol layer to network layer.}

Introduce simplified notation: \(z_P:=e_P(\Xi_P|_{[0,T]})\), \(z_N:=e_N(\Xi_N|_{[0,T]})\) as trajectory representations in unified space \(\mathcal{Z}\).

By condition (i) mapping non-degeneracy, for semantic pair \(x\neq x'\in\mathcal{X}\),
\begin{equation}
\mathbb{E}\big[d(z_P, z_N)\,\big|\,X=x\big] \le C,\quad
\mathbb{E}\big[d(z_P, z_N)\,\big|\,X=x'\big] \le C.
\end{equation}

By condition (iii) Lipschitz robustness, for any trajectories \(z_P, z_N\in\mathcal{Z}\),
\begin{equation}
|\varphi(z_P) - \varphi(z_N)| \le L_\varphi \cdot d(z_P, z_N).
\end{equation}

Taking conditional expectation with respect to \(X=x\):
\begin{align}
\big|\mathbb{E}[\varphi(z_P)\mid X=x] - \mathbb{E}[\varphi(z_N)\mid X=x]\big|
&= \big|\mathbb{E}[\varphi(z_P) - \varphi(z_N)\mid X=x]\big|\nonumber \\
&\le \mathbb{E}\big[|\varphi(z_P) - \varphi(z_N)|\,\big|\,X=x\big] \quad\text{(Jensen's inequality)}\nonumber\\
&\le \mathbb{E}\big[L_\varphi \cdot d(z_P,z_N)\,\big|\,X=x\big] \quad\text{(Lipschitz property)}\nonumber\\
&= L_\varphi \cdot \mathbb{E}\big[d(z_P,z_N)\,\big|\,X=x\big]\nonumber\\
&\le L_\varphi \cdot C. \quad\text{(Condition (i))}
\end{align}

Similarly for \(X=x'\),
\begin{equation}
\big|\mathbb{E}[\varphi(z_P)\mid X=x'] - \mathbb{E}[\varphi(z_N)\mid X=x']\big| \le L_\varphi \cdot C.
\end{equation}

Applying the triangle inequality:
\begin{align}
&\Big|\mathbb{E}[\varphi(z_N)\mid X=x]-\mathbb{E}[\varphi(z_N)\mid X=x']\Big| \nonumber\\
&\ge \Big|\mathbb{E}[\varphi(z_P)\mid X=x]-\mathbb{E}[\varphi(z_P)\mid X=x']\Big| \nonumber\\
&\quad - \Big|\mathbb{E}[\varphi(z_P)\mid X=x]-\mathbb{E}[\varphi(z_N)\mid X=x]\Big| \nonumber\\
&\quad - \Big|\mathbb{E}[\varphi(z_P)\mid X=x']-\mathbb{E}[\varphi(z_N)\mid X=x']\Big| \nonumber\\
&\ge \bar\Delta - L_\varphi C - L_\varphi C \nonumber\\
&= \bar\Delta - 2L_\varphi C\ =:\ \delta_N.
\end{align}

By condition (v), \(C < \frac{\bar\Delta}{2L_\varphi}\), thus \(\delta_N = \bar\Delta - 2L_\varphi C > 0\).

\textbf{Step 3: Propagation of expectation difference from network layer to observation layer.}

By condition (iv) non-degenerate observation, applying to the expectation difference \(\delta_N>0\) obtained in the previous step and semantic pair \((x,x')\), there exists bounded observation statistic \(\psi:\mathcal{Y}\to[-1,1]\) such that
\begin{equation}
\Big|\mathbb{E}[\psi(Y)\mid X=x]-\mathbb{E}[\psi(Y)\mid X=x']\Big|\ \ge\ \rho\,\delta_N
= \rho(\bar\Delta - 2L_\varphi C).
\end{equation}

\textbf{Step 4: From expectation difference to mutual information (including prior weighting).}

By Lemma \ref{lem:exp-to-tv}, applied to observation-layer conditional distributions and statistic \(\psi:\mathcal{Y}\to[-1,1]\) (i.e., \(M=1\)), we obtain
\begin{equation}
\mathrm{TV}\!\left(P_{Y\mid X=x},P_{Y\mid X=x'}\right)\ \ge\ \frac{\rho\,\delta_N}{2}
= \frac{\rho(\bar\Delta - 2L_\varphi C)}{2}.
\end{equation}

For general prior distributions, using the standard relationship between mutual information and total variation of conditional distributions (see \cite{cover2006elements}),
\begin{equation}
I(X;Y)\ \ge\ \frac{2}{\ln 2}\,\mathbb{P}(X=x)\mathbb{P}(X=x')\,\mathrm{TV}^2\!\left(P_{Y\mid X=x},P_{Y\mid X=x'}\right).
\end{equation}

Substituting the total variation lower bound:
\begin{equation}
I(X;Y)\ \ge\ \frac{2}{\ln 2}\,\mathbb{P}(X=x)\mathbb{P}(X=x')\left(\frac{\rho(\bar\Delta - 2L_\varphi C)}{2}\right)^2.
\end{equation}

By condition (v) and the prior positive mass condition, the right-hand side is strictly positive, thus \(I(X;Y) > 0\).

In the binary equal-prior subproblem (\(\mathbb{P}(X=x)=\mathbb{P}(X=x')=1/2\)), the above becomes
\begin{equation}
I(X;Y)\ \ge\ \frac{2}{\ln 2}\cdot\frac{1}{4}\left(\frac{\rho(\bar\Delta - 2L_\varphi C)}{2}\right)^2
= \frac{1}{2\ln 2}\left(\frac{\rho(\bar\Delta - 2L_\varphi C)}{2}\right)^2.
\end{equation}

This completes the proof.
\end{proof}

The mutual information lower bound is prior-dependent: the lower bound of \(I(X;Y)\) is weighted by the prior mass \(\mathbb{P}(X=x)\mathbb{P}(X=x')\) of the distinguishable semantic pair. When the prior is unknown, a conservative lower bound \(p_{\min}^2\) can be used, where \(p_{\min}\) is the minimum prior mass on the support set. The core insight of the proof is: conditions (i)-(v) constitute a stable propagation chain of expectation differences—distinguishability is propagated at each step with controllable loss \(L_\varphi C\); as long as condition (v) ensures that the total loss does not exceed half of the initial margin \(\bar\Delta\), it ultimately remains positive at the observation layer.

\subsection{Inevitability of Leakage in Multi-Semantic Spaces}

Based on the binary theorem, we now generalize to the general case of multi-semantic spaces, establishing the inevitability of side-channel leakage.

\begin{corollary}[Multi-Semantic Side-Channel Existence]\label{cor:multiclass-leakage}
Let semantic space \(\mathcal{X}\) be non-trivial (\(|\mathcal{X}|\ge 2\) and prior support contains at least two elements). Under the conditions of Proposition \ref{prop:kernel}, fix window \(T>0\). If encrypted communication system \(\Gamma\) and observation model \(\Omega\) satisfy:

\begin{enumerate}[label=(\roman*)]
\item \textbf{Efficiency-First Design}: there exist metric \(d\) and constant \(C<\infty\) such that for all \(x\in\mathcal{X}\),
\begin{equation}
\mathbb{E}\Big[d\big(e_P(\Xi_P|_{[0,T]}),\,e_N(\Xi_N|_{[0,T]})\big)\,\Big|\,X=x\Big]\ \le\ C.
\end{equation}

\item \textbf{Semantic Diversity}: there exists at least one distinguishable semantic pair \(x\neq x'\in\mathcal{X}\) and bounded Lipschitz statistic \(\varphi:\mathcal{Z}\to[-M,M]\) (with Lipschitz constant \(L_\varphi\)) such that
\begin{equation}
\big|\mathbb{E}[\varphi(e_P(\Xi_P|_{[0,T]}))\mid X=x]-\mathbb{E}[\varphi(e_P(\Xi_P|_{[0,T]}))\mid X=x']\big|\ \ge\ \bar\Delta>0,
\end{equation}
and \(\mathbb{P}(X=x)>0\), \(\mathbb{P}(X=x')>0\).

\item \textbf{Rational Observer}: observation model \(\Omega\) satisfies non-degeneracy (Definition \ref{def:nondegenerate}) for this statistic \(\varphi\), with \(\rho\in(0,1]\).

\item \textbf{Distinguishability Propagation Condition}: \(C < \frac{\bar\Delta}{2L_\varphi}\)
\end{enumerate}

Then the mutual information between observed features and semantics satisfies \(I(X;Y) > 0\).
\end{corollary}

\begin{proof}
By condition (i), efficiency-first design holds for all semantics, in particular for semantic pair \((x,x')\). Combined with conditions (ii)(iii)(iv), this semantic pair satisfies all conditions of Theorem \ref{thm:binary-leakage}. Therefore
\begin{equation}
I(X;Y) \ge \frac{2}{\ln 2}\,\mathbb{P}(X=x)\mathbb{P}(X=x')\,
\left(\frac{\rho\,[\bar\Delta-2L_\varphi C]}{2}\right)^2 > 0.
\end{equation}
This completes the proof.
\end{proof}

Corollary \ref{cor:multiclass-leakage} shows: in efficiency-prioritized multi-semantic systems, as long as at least one pair of applications is statistically distinguishable, side-channel leakage is inevitable. The universality of this conclusion stems from:

(1) \textbf{Efficiency-first is a system-level constraint} (condition i): real-world systems must satisfy bandwidth, latency, and other performance requirements, thus \(C<\infty\) for all semantics.

(2) \textbf{Semantic diversity is an inevitable consequence of applications} (condition ii): different application types—such as video streaming (large and dense packets), web browsing (small and sparse packets), instant messaging (bidirectional and symmetric), file transfer (unidirectional and concentrated)—necessarily exhibit differences in statistical features such as packet size distributions, timing patterns, and upstream/downstream ratios. This stems from application logic itself and is independent of encryption.

(3) \textbf{Rational observer is the analyst's objective} (condition iii): side-channel analysts aim to maximize information extraction and will retain key statistical features, thus \(\rho>0\).

Condition (ii) only requires ``at least one distinguishable pair exists'', not ``all semantics are pairwise distinguishable''—this is an extremely weak assumption. In real-world systems containing \(n\ge 2\) applications, it is nearly impossible to make all applications produce statistically indistinguishable traffic—this would require fundamentally changing how applications work, contradicting the purpose of application design.

Therefore, in the universal situation of ``efficiency-prioritized usable systems + non-trivial application scenarios + rational observers'', side-channel leakage \(I(X;Y)>0\) is inevitable.

\section{Theoretical Analysis and Discussion}

This section interprets the operational meaning of the existence theorem, discusses the transformation from information-theoretic lower bounds to actual attack performance, and the efficiency-privacy tradeoff revealed by the theorem.

\subsection{Operational Interpretation from Information-Theoretic Lower Bounds to Attack Feasibility}

Theorem \ref{thm:binary-leakage} and Corollary \ref{cor:multiclass-leakage} assert that \(I(X;Y)>0\) in efficiency-prioritized systems and provide explicit lower bounds. To translate this information-theoretic statement into operational predictions for actual attack performance, we establish a precise connection from total variation to classification accuracy.

\textbf{Accuracy lower bound in binary case.}
For the binary equal-prior problem (\(\mathbb{P}(X=x)=\mathbb{P}(X=x')=1/2\)), the error rate and accuracy of the optimal Bayes classifier satisfy the precise relationship:
\begin{equation}
P_e^\star = \frac{1-\mathrm{TV}(P_{Y\mid X=x},P_{Y\mid X=x'})}{2},\quad
\text{Acc}^\star = \frac{1+\mathrm{TV}(P_{Y\mid X=x},P_{Y\mid X=x'})}{2}.
\end{equation}

From step 4 of Theorem \ref{thm:binary-leakage}, we already obtained the total variation lower bound
\begin{equation}
\mathrm{TV}(P_{Y\mid X=x},P_{Y\mid X=x'}) \ge \frac{\rho(\bar\Delta-2L_\varphi C)}{2},
\end{equation}
where \(C\) is the mapping non-degeneracy constant (theorem condition (i)).

Substituting into the above, we obtain a lower bound on optimal accuracy:
\begin{equation}
\text{Acc}^\star \ge \min\left\{1,\ \frac{1}{2} + \frac{1}{4}\rho(\bar\Delta-2L_\varphi C)\right\}.
\end{equation}
Since statistic \(\varphi:\mathcal{Z}\to[-M,M]\) is bounded, \(\bar\Delta\le 2M\); combined with the setting of \(\rho\in(0,1]\) and \(\psi:\mathcal{Y}\to[-1,1]\) in the non-degenerate observation definition, the right-hand side naturally stays within bounds.

This is a lower bound on the optimal Bayes classifier accuracy, implying there exists a classifier achieving this level. For example, if \(\rho=0.8\), \(\bar\Delta=1.0\), \(L_\varphi C=0.2\), then
\begin{equation}
\text{Acc}^\star \ge \frac{1}{2} + \frac{1}{4}\times 0.8\times(1.0-0.4) = 0.62,
\end{equation}
meaning optimal accuracy is at least 62\%.

For general priors or multi-class cases (\(M>2\)), one can apply this bound to the hardest-to-distinguish binary subproblem (the bound may be loose), or use ``one-vs-all'' union strategies to obtain conservative lower bounds.

\textbf{Error rate lower bound from Fano's inequality.}
As a complement, we can also use Fano's inequality to characterize the information-theoretic error rate lower bound. Let semantic space \(\mathcal{X}\) contain \(M\ge 2\) elements, and \(P_e^\star\) be the minimum Bayes error probability. The classical Fano inequality gives
\begin{equation}
H(X|Y) \le H_2(P_e^\star) + P_e^\star\log_2(M-1),
\end{equation}
where \(H_2(p)=-p\log_2 p-(1-p)\log_2(1-p)\) is binary entropy (in bits), and \(H(X|Y)=H(X)-I(X;Y)\) is posterior entropy. Rearranging gives
\begin{equation}
P_e^\star \ge \frac{H(X)-I(X;Y)-1}{\log_2(M-1)}.
\end{equation}

For the binary equal-prior case, \(H(X)=1\), the above is equivalent to \(I(X;Y) \ge 1 - H_2(P_e^\star)\), thus \(P_e^\star \ge H_2^{-1}(1-I(X;Y))\).

For example, if \(I(X;Y)=0.1\) bits, then \(P_e^\star \ge H_2^{-1}(0.9)\approx 0.317\), meaning optimal error rate is at least 31.7\%, and accuracy at most approximately 68.3\%. Note that Fano's inequality provides a lower bound on error rate rather than an upper bound; it clarifies ``the information-theoretic limit that even optimal classifiers cannot surpass'', but cannot directly predict achievable performance of actual attacks.

\textbf{Cumulative effect of multiple observations.}
In actual side-channel analysis, attackers can often observe multiple sessions \(Y^{(1)}, Y^{(2)}, \ldots, Y^{(n)}\) (different sessions of the same user visiting the same website, or concatenating multiple time windows of the same flow). Assuming that given semantic \(X\), each observation is independent and identically distributed:
\begin{equation}
Y^{(1)}, \ldots, Y^{(n)} \overset{\text{i.i.d.}}{\sim} P_{Y\mid X}\quad\text{given}\quad X,
\end{equation}
then joint mutual information satisfies additivity:
\begin{equation}
I(X; Y^{(1:n)}) = \sum_{i=1}^n I(X; Y^{(i)}) = n \cdot I(X;Y).
\end{equation}

In this case, the error rate decays exponentially with the Chernoff information as exponent. For the binary equal-prior problem, the large deviation asymptotics of optimal Bayes error rate are characterized by the exact theorem:
\begin{equation}
-\lim_{n\to\infty}\frac{1}{n}\log P_e^\star(n) = \mathcal{C}(P_{Y\mid X=x}, P_{Y\mid X=x'}),
\end{equation}
where \(\mathcal{C}(\cdot,\cdot)\) is the Chernoff information (here logarithm base is \(e\), unit is nats), and \(\log\) denotes natural logarithm.

To establish a lower bound chain from total variation to Chernoff information, we introduce the Bhattacharyya coefficient \(\mathrm{BC}=\int\sqrt{p(y)q(y)}\,\mathrm{d}y\) and Bhattacharyya distance \(B=-\ln\mathrm{BC}\). Known relationships are
\begin{equation}
\mathrm{TV}(P,Q) \le \sqrt{1-e^{-2B}},\quad \mathcal{C}(P,Q) \ge B,
\end{equation}
thus
\begin{equation}
\mathcal{C}(P,Q) \ge B \ge -\frac{1}{2}\ln\!\big(1-\mathrm{TV}^2(P,Q)\big) > 0.
\end{equation}

Substituting our theorem result \(\mathrm{TV}(P_{Y\mid X=x},P_{Y\mid X=x'}) \ge \frac{\rho(\bar\Delta-2L_\varphi C)}{2}\) into the above gives an explicit lower bound on Chernoff information:
\begin{equation}
\mathcal{C}(P_{Y\mid X=x}, P_{Y\mid X=x'}) \ge -\frac{1}{2}\ln\!\left(1-\left[\frac{\rho(\bar\Delta-2L_\varphi C)}{2}\right]^2\right) > 0.
\end{equation}

This guarantees that the error rate exponentially tends to zero with the number of observations \(n\).

For the multi-class case (\(M>2\)), if each class has positive prior mass and samples are independent and identically distributed given \(X\), the error exponent lower bound of overall Bayes error rate is controlled by \(\min_{x\neq x'} \mathcal{C}(P_{Y\mid X=x}, P_{Y\mid X=x'})\). This conclusion comes from common union and worst-pair domination arguments: overall error rate is controlled by the hardest-to-distinguish semantic pair.

This explains universally observed phenomena in practice:

(1) \textbf{Long observation windows improve accuracy}: When growth of \(\bar\Delta(T)\) is not offset by \(2L_\varphi C(T)\), increasing time window \(T\) makes \(\delta_N(T)=\bar\Delta(T)-2L_\varphi C(T)\) increase, thus total variation lower bound and accuracy lower bound increase with \(T\).

(2) \textbf{Concatenating multiple sessions significantly improves identification}: Concatenating \(n\) independent sessions makes mutual information accumulate linearly to \(n\cdot I(X;Y)\) (in bits), with accuracy converging at Chernoff exponent (in nats).

(3) \textbf{Exponential convergence to perfect identification}: Under the conditional independence assumption, error rate decays exponentially as \(\exp(-n\cdot \mathcal{C})\) to zero.

Corollary \ref{cor:multiclass-leakage} guarantees \(I(X;Y)>0\) necessarily holds in efficiency-prioritized systems, thus the above cumulative effects are inevitable—as long as attackers have sufficient observation budget and the conditional independence assumption is satisfied, identification accuracy will tend toward perfection. This is the operational meaning of ``ineliminability'' of side-channel leakage.

\subsection{Fundamentality and Insurmountability of Efficiency-Privacy Tradeoff}

The five conditions of Theorem \ref{thm:binary-leakage} reveal the only way to reduce leakage \(I(X;Y)\) and its costs. We analyze the ``cost of breaking'' each condition.

\textbf{Condition (i): Mapping non-degeneracy \(C<\infty\).}
This is a direct manifestation of efficiency-first design. To break this condition (increase \(C\to\infty\)), the system must pay costs in at least one of the following dimensions:

Length dimension: Heavy padding increases \(\mathbb{E}[d_{\text{length}}(z_P, z_N)]\). For example, padding all packets to MTU (1500 bytes) inflates small packets (such as 40-byte ACKs) by tens of times, with bandwidth overhead reaching order-of-magnitude growth.

Timing dimension: Artificial delays increase \(\mathbb{E}[d_{\text{time}}(z_P, z_N)]\). For example, introducing second-level delays breaks real-time applications (VoIP requires one-way latency \(<\)150ms).

Direction dimension: Cover traffic changes upstream/downstream ratios, increasing \(\mathbb{E}[d_{\text{direction}}(z_P, z_N)]\). Bidirectional cover doubles bandwidth overhead.

Condition (v) in the theorem requires \(C < \frac{\bar\Delta}{2L_\varphi}\) to ensure leakage propagation. To make the leakage lower bound tend to zero, we need \(C\to\frac{\bar\Delta}{2L_\varphi}\), meaning efficiency overhead tends toward a critical value—this critical value is determined by the inherent distinguishability \(\bar\Delta\) of applications and cannot be changed.

\textbf{Condition (ii): Semantic distinguishability \(\bar\Delta>0\).}
This is an inevitable consequence of application diversity. To break this condition (make \(\bar\Delta\to 0\)), all applications need to produce statistically indistinguishable traffic. This is nearly impossible in practice; for instance, video streaming and web browsing differ by orders of magnitude in bandwidth requirements. Instant messaging and file downloading are fundamentally different in interaction patterns. VoIP and HTTP show significant differences in timing features.

To completely eliminate \(\bar\Delta\), all applications must be forced to transmit at the same constant rate, same packet size, same bidirectional pattern, which thoroughly destroys application functionality and makes the problem meaningless.

\textbf{Conditions (iii)-(iv): Lipschitz property and observation non-degeneracy.}
Condition (iii) is a technical requirement for statistic selection; in practice, almost all useful statistics (window total bytes, packet counts, upstream/downstream ratios, etc.) can satisfy Lipschitz property through truncation or normalization. Condition (iv) characterizes rational observers, determined by observer technical capabilities rather than controllable by system designers.

\textbf{Insurmountability of the tradeoff.}
Synthesizing the above analysis, the cost of reducing leakage presents a trilemma: when the broken condition is increasing \(C\) (relaxing non-degeneracy), the cost paid is sacrificing efficiency (bandwidth/latency); when the broken condition is decreasing \(\bar\Delta\) (homogenizing applications), the cost paid is destroying functionality (applications unusable); when the broken condition is decreasing \(\rho\) (compressing observation), the cost paid is beyond control (determined by observer).

Under fixed business requirement constraints such as bandwidth overhead requirements and end-to-end latency requirements, there exists an insurmountable leakage lower bound. This is not a flaw of any particular protocol implementation, but a structural limitation jointly determined by efficiency priority and semantic diversity.

\subsection{Theoretical Boundaries of Defense and Correct Engineering Objectives}

Corollary \ref{cor:multiclass-leakage} shows that zero leakage \((I(X;Y)=0)\) is unattainable in efficiency-prioritized systems. This conclusion has important guiding significance for defense mechanism design. Theorem \ref{thm:binary-leakage} further reveals the mechanism of action of defense mechanisms: constant-rate padding increases metric deviation \(C\) to blur the mapping from protocol layer to network layer, at the cost of bandwidth overhead; differential privacy mechanisms reduce observation fidelity by adding noise, at the cost of utility loss; while complete obfuscation defenses attempt to reduce semantic distinguishability \(\bar\Delta\), but will destroy application functionality. The essence of these mechanisms is finding different tradeoff points in the constraint space formed by the trilemma.

\textbf{Wrong objective: pursuing \(I(X;Y)=0\).}
Many defense schemes (such as Tor's traffic obfuscation, VPN's constant-rate padding, etc.) implicitly aim for ``eliminating side channels''. The theorem shows this objective is unattainable while maintaining system usability. Phenomena observed in practice verify this conclusion:

(1) Strong defenses like Tamaraw increase latency by 78\% and bandwidth overhead to 135\% in real Tor network deployment\cite{meng2024palette}; although they can significantly reduce attack accuracy, their high overhead limits practical deployment.

(2) Constant-rate padding strategies like BuFLO and CS-BuFLO can reduce identification accuracy but require over 100\% bandwidth overhead\cite{dyer2012peek,cai2014csbuflo}, making them impractical for actual deployment.

(3) Obfuscation defenses like WTF-PAD and FRONT have lower overhead but cannot resist latest deep learning attacks (accuracy can reach over 90\%)\cite{meng2024palette}.

\textbf{Correct objective: constrained optimization.}
The correct engineering objective revealed by the theorem is: minimize leakage under given efficiency constraints and functional requirements. Formalized as a constrained optimization problem:
\begin{equation}
\min_{\theta\in\Theta} I(X;Y;\theta) \quad\text{s.t.}\quad 
\begin{cases}
\text{Bandwidth overhead} \le \beta_{\max} & \text{(e.g., 10\%)}\\
\text{Latency increase} \le \Delta t_{\max} & \text{(e.g., 50ms)}\\
\text{Application functionality complete} & \text{(\(\bar\Delta\ge\Delta_{\min}\))}
\end{cases}
\end{equation}
where \(\theta\) is a joint vector of defense parameters such as padding strategies, timing perturbation, and cover traffic.

\section{Conclusion}

This paper provides a formal model \(\Sigma=(\Gamma,\Omega)\) for side-channel analysis from information theory and system design, abstracting the entire process of ``generation, encapsulation, encryption, transmission, observation'' as a causally measurable Markov chain \(X\to\Xi_A\to\Xi_P\to\Xi_C\to\Xi_N\to Y\). Based on this framework, this paper proves the side-channel existence theorem (Theorem \ref{thm:binary-leakage}): for distinguishable binary semantic pairs, under the conditions of mapping non-degeneracy (\(\mathbb{E}[d(z_P,z_N)\mid X]\le C\)), semantic distinguishability (expectation difference \(\ge\bar\Delta\)), Lipschitz robustness (\(|\varphi(z)-\varphi(z')|\le L_\varphi d(z,z')\)), non-degenerate observation (preservation ratio \(\rho>0\)), and the distinguishability propagation condition (\(C<\bar\Delta/2L_\varphi\)), the mutual information between observed features and semantic variables satisfies the explicit lower bound \(I(X;Y)\ge\frac{1}{2\ln 2}\left(\frac{\rho[\bar\Delta-2L_\varphi C]}{2}\right)^2>0\). Corollary \ref{cor:multiclass-leakage} further shows: in efficiency-prioritized multi-semantic systems, as long as at least one pair of applications is statistically distinguishable (stemming from inherent differences in application logic), side-channel leakage is inevitable. Through the precise relationship between total variation and accuracy and the Bhattacharyya-Chernoff lower bound chain, this paper establishes a quantified connection from information-theoretic lower bounds to actual attack performance, revealing the inevitability that multiple observations exponentially accumulate under the conditional independence assumption, making identification accuracy tend toward perfection.

This paper's analysis shows that reducing leakage faces a trilemma: increasing metric deviation \(C\) requires sacrificing efficiency, decreasing semantic distinguishability \(\bar\Delta\) will destroy application functionality, while observation non-degeneracy \(\rho\) is controlled by analysts rather than system designers. Therefore, the correct engineering objective is not pursuing unattainable zero leakage, but rather a constrained optimization problem that minimizes leakage under given efficiency constraints and functional requirements. It is worth emphasizing that the existence theorem and explicit mutual information lower bounds provided in this paper can serve as a theoretical baseline for protocol evolution and defense evaluation: under given efficiency constraints such as bandwidth, latency, and compatibility, different protocol configurations or defense strategies can be mapped to changes in metric \(d\) and non-degeneracy constant \(C(T)\), thereby quantifying their impact on leakage lower bounds and conducting goal-oriented constrained optimization and scheme comparison.

Particularly in website fingerprinting scenarios, this paper's theoretical framework provides direct guidance for attack feature selection. According to the structure of the mutual information lower bound, attackers should prioritize selecting macroscopic aggregate features with large expectation differences \(\bar\Delta\) (such as session total bytes, upstream/downstream packet ratios, burst payloads, etc.), while considering Lipschitz robustness of statistics to resist network perturbations; they should value distinguishability information in timing structures, such as packet interval distributions and burst-silence alternation patterns; they can also exploit the Chernoff exponential effect of multi-session observations to achieve exponential-level improvement in identification accuracy. These insights show that the expectation difference—Lipschitz robustness—mapping non-degeneracy analysis chain established in this paper provides an interpretable theoretical basis and operational feature design guidelines for traffic feature side-channel attacks.

Several open problems remain in the practical implementation of this paper's theoretical framework. First, the lower bounds provided by the theorem depend on the choice of metric \(d\) and estimation of constants \(C\), \(L_\varphi\); how to identify or verify these parameters from measured traffic data, and how to establish parameter libraries for different protocol families (TLS 1.3, QUIC) and business types (video streaming, web browsing), are key steps for theory landing. Second, this paper's non-degeneracy conditions (Definitions \ref{def:non-degeneracy}, \ref{def:nondegenerate}) are based on metric bounds in the expectation sense, but real networks have transient perturbations such as congestion bursts and routing jitter; how to reformulate conditions in a probabilistic sense (such as high-probability bounds or quantile constraints) and derive corresponding leakage lower bounds will enhance conclusion robustness. Third, this paper focuses on passive observation scenarios, but active probing (such as induced visits in website fingerprinting attacks) and adaptive attacks (observers adjusting strategies based on intermediate results) may break static lower bounds; how to characterize Nash equilibria and optimal strategies of attack-defense parties in a game-theoretic framework is an important topic in dynamic adversarial environments. Additionally, the precise relationship between this paper's information-theoretic mutual information measure and differential privacy's \((\varepsilon,\delta)\)-DP guarantees has not yet been established; whether one can prove ``mechanisms satisfying \(\varepsilon\)-DP necessarily lead to mutual information dropping below \(f(\varepsilon)\)'' or similar equivalence theorems will provide operational formal guidelines for privacy mechanism design. Finally, hierarchical leakage analysis in multi-task scenarios (such as first identifying application categories then subdividing specific websites), and dynamic leakage accumulation models considering temporal correlations and long-term observations, are all directions worthy of in-depth research. Solving these problems will advance this paper's existence conclusions into a computable, verifiable, and optimizable engineering practice framework.


\begin{thebibliography}{99}

\bibitem{wang2014effective}
WANG T, GOLDBERG I. Effective attacks and provable defenses for website fingerprinting[C]//Proceedings of the 23rd USENIX Security Symposium. Berkeley: USENIX Association, 2014: 143-157.

\bibitem{shen2021graphdapp}
SHEN M, ZHANG J, ZHU L, et al. Accurate decentralized application identification via encrypted traffic analysis using graph neural networks[J]. IEEE Transactions on Information Forensics and Security, 2021, 16: 2367-2380.

\bibitem{mei2025high}
MEI H, CHENG G, YUAN Y. High precision and efficient anonymous traffic classification in the real-world[J]. IEEE Transactions on Networking, 2025, 33(2): 1256-1270.

\bibitem{kocher1996timing}
KOCHER P C. Timing attacks on implementations of Diffie-Hellman, RSA, DSS, and other systems[C]//Advances in Cryptology - CRYPTO '96: 16th Annual International Cryptology Conference. Berlin: Springer-Verlag, 1996: 104-113.

\bibitem{kocher1999differential}
KOCHER P, JAFFE J, JUN B. Differential power analysis[C]//Advances in Cryptology - CRYPTO '99: 19th Annual International Cryptology Conference. Berlin: Springer-Verlag, 1999: 388-397.

\bibitem{hintz2002fingerprinting}
HINTZ A. Fingerprinting websites using traffic analysis[C]//Privacy Enhancing Technologies: Second International Workshop, PET 2002. Berlin: Springer-Verlag, 2003: 171-178.

\bibitem{lin2022etbert}
LIN X, XIONG G, GOU G, et al. ET-BERT: a contextualized datagram representation with pre-training transformers for encrypted traffic classification[C]//Proceedings of the ACM Web Conference 2022. New York: ACM, 2022: 633-642.

\bibitem{shen2023machine}
SHEN M, YE K, LIU X, et al. Machine learning-powered encrypted network traffic analysis: a comprehensive survey[J]. IEEE Communications Surveys \& Tutorials, 2023, 25(1): 791-824.

\bibitem{li2018measuring}
LI S, GUO H, HOPPER N. Measuring information leakage in website fingerprinting attacks and defenses[C]//Proceedings of the 2018 ACM SIGSAC Conference on Computer and Communications Security. New York: ACM, 2018: 1977-1992.

\bibitem{dwork2006distributed}
DWORK C, KENTHAPADI K, MCSHERRY F, et al. Our data, ourselves: privacy via distributed noise generation[C]//Advances in Cryptology - EUROCRYPT 2006: 25th Annual International Conference on the Theory and Applications of Cryptographic Techniques. Berlin: Springer-Verlag, 2006: 486-503.

\bibitem{sabzi2024netshaper}
SABZI A, VORA R, GOSWAMI S, et al. NetShaper: a differentially private network side-channel mitigation system[C]//Proceedings of the 33rd USENIX Security Symposium. Berkeley: USENIX Association, 2024: 3385-3402.

\bibitem{chaum1981untraceable}
CHAUM D L. Untraceable electronic mail, return addresses, and digital pseudonyms[J]. Communications of the ACM, 1981, 24(2): 84-90.

\bibitem{diaz2002measuring}
DÍAZ C, SEYS S, CLAESSENS J, et al. Towards measuring anonymity[C]//Privacy Enhancing Technologies: Second International Workshop, PET 2002. Berlin: Springer-Verlag, 2003: 54-68.

\bibitem{deng2006measuring}
DENG Y, PANG J, WU P. Measuring anonymity with relative entropy[C]//Formal Aspects in Security and Trust: 4th International Workshop, FAST 2006. Berlin: Springer-Verlag, 2007: 65-79.

\bibitem{serjantov2002information}
SERJANTOV A, DANEZIS G. Towards an information theoretic metric for anonymity[C]//Privacy Enhancing Technologies: Second International Workshop, PET 2002. Berlin: Springer-Verlag, 2003: 41-53.

\bibitem{chatzikokolakis2007anonymity}
CHATZIKOKOLAKIS K, PALAMIDESSI C, PANANGADEN P. Anonymity protocols as noisy channels[J]. Information and Computation, 2008, 206(2-4): 378-401.

\bibitem{kesdogan2002limits}
KESDOGAN D, AGRAWAL D, PENZ S. Limits of anonymity in open environments[C]//Information Hiding: 5th International Workshop, IH 2002. Berlin: Springer-Verlag, 2003: 53-69.

\bibitem{danezis2003statistical}
DANEZIS G. Statistical disclosure attacks: traffic confirmation in open environments[C]//Security and Privacy in the Age of Uncertainty: IFIP TC11 18th International Conference on Information Security. Boston: Springer, 2003: 421-426.

\bibitem{cover2006elements}
COVER T M, THOMAS J A. Elements of information theory[M]. 2nd ed. Hoboken: John Wiley \& Sons, 2006: 34-35.

\bibitem{vandenhooff2015vuvuzela}
VAN DEN HOOFF J, LAZAR D, ZAHARIA M, et al. Vuvuzela: scalable private messaging resistant to traffic analysis[C]//Proceedings of the 25th Symposium on Operating Systems Principles. New York: ACM, 2015: 137-152.

\bibitem{tyagi2017stadium}
TYAGI N, GILAD Y, LEUNG D, et al. Stadium: a distributed metadata-private messaging system[C]//Proceedings of the 26th Symposium on Operating Systems Principles. New York: ACM, 2017: 423-440.

\bibitem{panchenko2011website}
PANCHENKO A, NIESSEN L, ZINNEN A, et al. Website fingerprinting in onion routing based anonymization networks[C]//Proceedings of the 10th Annual ACM Workshop on Privacy in the Electronic Society. New York: ACM, 2011: 103-114.

\bibitem{dyer2012peek} 
DYER K P, COULL S E, RISTENPART T, et al. Peek-a-boo, I still see you: Why efficient traffic analysis countermeasures fail[C]//Proceedings of the 2012 IEEE Symposium on Security and Privacy. Los Alamitos: IEEE Computer Society, 2012: 332-346.

\bibitem{cai2014systematic} 
CAI X, NITHYANAND R, WANG T, et al. A systematic approach to developing and evaluating website fingerprinting defenses[C]//Proceedings of the 2014 ACM SIGSAC Conference on Computer and Communications Security. New York: ACM, 2014: 227-238.

\bibitem{wang2017walkie} 
WANG T, GOLDBERG I. Walkie-talkie: An efficient defense against passive website fingerprinting attacks[C]//Proceedings of the 26th USENIX Security Symposium. Berkeley: USENIX Association, 2017: 1375-1390.

\bibitem{huang2025stap}
HUANG J N, LIU W, LIU G, et al. STAP: leveraging state-transition adversarial perturbations for asymmetric website fingerprinting defenses[J]. IEEE Transactions on Network and Service Management, 2025, 22(1): 234-248.

\bibitem{wright2009traffic}
WRIGHT C V, COULL S E, MONROSE F. Traffic morphing: an efficient defense against statistical traffic analysis[C]//Proceedings of the 16th Network and Distributed Security Symposium. Reston: The Internet Society, 2009: 237-250.

\bibitem{cherubin2017bayes}
CHERUBIN G. Bayes, not naïve: security bounds on website fingerprinting defenses[C]//Proceedings on Privacy Enhancing Technologies, 2017, 2017(4): 215-231.

\bibitem{fu2021realtime}
FU C, LI Q, SHEN M, et al. Realtime robust malicious traffic detection via frequency domain analysis[C]//Proceedings of the 2021 ACM SIGSAC Conference on Computer and Communications Security. New York: ACM, 2021: 3431-3446.

\bibitem{fu2024detecting}
FU C, LI Q, SHEN M, et al. Detecting tunneled flooding traffic via deep semantic analysis of packet length patterns[C]//Proceedings of the 2024 on ACM SIGSAC Conference on Computer and Communications Security. New York: ACM, 2024: 3659-3673.

\bibitem{fu2023detecting}
FU C P, LI Q, XU K. Detecting unknown encrypted malicious traffic in real time via flow interaction graph analysis[C]//Network and Distributed System Security Symposium (NDSS). San Diego, CA, USA: The Internet Society, 2023.

\bibitem{camenisch2005formal}
CAMENISCH J, LYSYANSKAYA A. A formal treatment of onion routing[C]//Advances in Cryptology - CRYPTO 2005: 25th Annual International Cryptology Conference. Berlin: Springer-Verlag, 2005: 169-187.

\bibitem{feigenbaum2007model}
FEIGENBAUM J, JOHNSON A, SYVERSON P. A model of onion routing with provable anonymity[C]//Financial Cryptography and Data Security: 11th International Conference, FC 2007. Berlin: Springer-Verlag, 2007: 57-71.

\bibitem{danezis2009sphinx}
DANEZIS G, GOLDBERG I. Sphinx: a compact and provably secure mix format[C]//Proceedings of the 30th IEEE Symposium on Security and Privacy. Los Alamitos: IEEE Computer Society, 2009: 269-282.

\bibitem{gray2011probability}
GRAY R M. Probability, random processes, and ergodic properties[M]. 2nd ed. New York: Springer, 2011.

\bibitem{meng2024palette}
SHEN M, XU K, LI Q, et al. Real-time website fingerprinting defense via traffic cluster anonymization[C]//Proceedings of the 2024 IEEE Symposium on Security and Privacy. Los Alamitos: IEEE Computer Society, 2024: 2674-2691.

\bibitem{cai2014csbuflo}
CAI X, NITHYANAND R, JOHNSON R. CS-BuFLO: a congestion sensitive website fingerprinting defense[C]//Proceedings of the 2014 ACM Workshop on Privacy in the Electronic Society. New York: ACM, 2014: 121-130.

\end{thebibliography}
\end{document}